\documentclass[onecolumn,journal]{IEEEtran}
\usepackage{amssymb,amsmath,amsthm}
\usepackage{graphicx} 
\usepackage{tikz}
\usetikzlibrary{positioning}
\usetikzlibrary{arrows,patterns,calc,graphs}
\usetikzlibrary{fit}
\title{A Note on Function Correcting Codes for b-Symbol Read Channels}
\author{Sachin Sampath and B. Sundar Rajan \\
Department of Electrical Communication Engineering\\
Indian Institute of Science, Bangalore}
\date{January 2025}

\newtheorem{definition}{Definition}

\newtheorem{theorem}{Theorem}
\newtheorem{corollary}{Corollary}

\usepackage{algorithm}
\usepackage{algpseudocode}

\begin{document}
\maketitle

\begin{abstract}
 Function-Correcting Codes (FCCs) is a novel paradigm in Error Control Coding introduced by Lenz et. al. 2023 for the binary substitution channel \cite{FCC}. FCCs aim to protect the function evaluation of data against errors instead of the data itself, thereby relaxing the redundancy requirements of the code. Later R. Premlal et. al. \cite{LFCC} gave new bounds on the optimal redundancy of FCCs and also extensively studied FCCs for linear functions. The notion of FCCs has also been extended to different channels such as symbol-pair read channel over the binary field by Xia et. al. \cite{FCSPC} and b-symbol read channel over finite fields by A.Singh et. al. \cite{FCBSC}. In this work, we study FCCs for linear functions for the b-symbol read channel. We provide the Plotkin-like bound on FCCs for b-symbol read channel which reduces to a Plotkin-like bound for FCCs for the symbol-pair read channel when $b$=2. FCCs reduce to classical Error Correcting Codes (ECCs) when the function is bijective. Analogous to this our bound reduces to the Plotkin-bound for classical ECCS for both the b-symbol and symbol-pair read channels \cite{Plotkin-b-symbol, Plotkin-symbol-pair} when we consider linear bijective functions.  
\end{abstract}

\section{Introduction}
In a standard communication system, the transmitter wishes to convey a message to the receiver through a noisy channel. Classical Error Correcting Codes (ECC) aim to counteract the noise introduced by the channel and recover the entire message correctly by adding redundancy to the message. Suppose the receiver is interested only in the function evaluation of the message instead of the message itself for some function. If the transmitter also knows the function then they can encode in such a way that only the function evaluations are protected, which relaxes the redundancy requirements of the code. This leads to a new class of codes called Function-Correcting Codes (FCC) introduced by Lenz et. al. in \cite{FCC}. The key idea here is that the receiver does not need to distinguish between different messages that evaluate to the same function value. Hence codewords assigned to these messages have no distance requirements between each other.  To this end \cite{FCC} defines FCCs, irregular distance codes, function-dependent graphs, etc which are then used to provide bounds on the optimal redundancy of FCCs. It also gives code constructions for several functions and compares the redundancy to the redundancy of classical ECCs over the data and also the function values.

Later works have extended the notion of FCCs for different channels \cite{FCSPC}, \cite{FCBSC} and for different functions \cite{LFCC}. Optimal and near optimal code constructions have also been given for some functions \cite{OFCC}. This work aims to study FCCs for b-symbol read channels for linear functions. We have derived a Plotkin-like bound for linear functions which gives the Plotkin-like bound for linear functions over symbol-pair read channels. Further restricting to linear bijective functions, we recover the Plotkin-like bounds on Classical ECCs over b-symbol read and symbol-pair read channels  \cite{Plotkin-b-symbol, Plotkin-symbol-pair} as special cases of our lower bound.
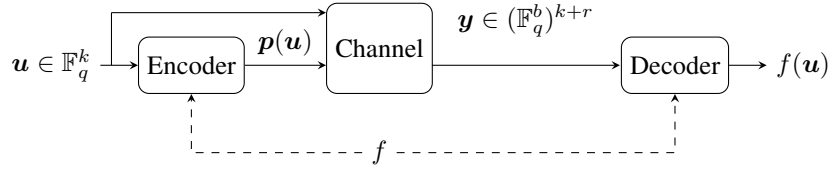
\begin{figure}[t]
	\centering
	\begin{tikzpicture}[>=stealth]
		\node (u) {$\boldsymbol{u} \in \mathbb{F}_q^k$};
		\node[right = 0.5cm of u, rectangle, rounded corners, draw, minimum height=0.75cm] (enc)              {Encoder};
		\node[draw,rectangle, rounded corners, right = 2.5 cm of enc.south, minimum height=1.2cm, anchor      = south] (channel) {Channel};
		\node[draw,rectangle, rounded corners, right = 5 cm of enc.east,minimum height=0.75cm] (decoder)      {Decoder};
		\node[below = 0.5cm of channel] (fu) {$f$};
		
		\node[right = 0.5cm of decoder](fy) {$f(\boldsymbol{u})$};
				
		\draw[->] (u) -- (enc);
		\coordinate[] (cor) at ($(u) + (0.8,0.7)$);
		\draw[-] (u -| cor) -- (cor);
		\draw[->] (cor) -- (cor -| channel.west);
		\draw[->] (enc.east) -- node [above] {$\boldsymbol{p}(\boldsymbol{u})$} (enc.east -| channel.west);
        \coordinate[] (cor2) at ($(u) + (5.06,0)$);
         \coordinate[] (cor3) at ($(u) + (7.56,0)$);
		\draw[->] (cor2) -- node [above = 0.25cm] {$\boldsymbol{y} \in (\mathbb{F}_q^b)^{k+r}$} (cor3);
		\draw[->] (decoder) -- (fy);
		\draw[->, dashed] (fu) -| (enc);
		\draw[->, dashed] (fu) -| (decoder);
	\end{tikzpicture}

	\caption{Illustration of the FCBSC setup. The input to the encoder is a message $\boldsymbol{u}$, which features an attribute $f(\boldsymbol{u})$ that is of special interest to receiver. To guarantee recoverability of this attribute, transmitter encodes the message $\boldsymbol{u}$ to a redundancy vector $\boldsymbol{p(u)}$ and transmits $(\boldsymbol{u},\boldsymbol{p(u)})$. Given an erroneous version $\boldsymbol{y}$ of the codeword $\boldsymbol{c}=(\boldsymbol{u},\boldsymbol{p})$ and the knowledge of the function $f$, the decoder can correctly infer $f(\boldsymbol{u})$.}
	\label{fig:intro_fcc_problem}
\end{figure}%
\subsection{Notation}
The set of all positive integers is denoted by $\mathbb{N}$ and the set of all non-negative integers is denoted by $\mathbb{N}_o$. The set of all positive integers less than or equal to $k$ is denoted by $[k]$. For a matrix $\boldsymbol{B}$, $[\boldsymbol{B}]_{ij}$ denotes the $(i,j)^{th}$ entry of $\boldsymbol{B}.$ For a vector $\boldsymbol{u}$ of length $n$, $\boldsymbol{u}_i$ denotes the $i^{th}$ entry of $\boldsymbol{u}.$ We write $a \equiv b(mod \ c)$ if $a-b$ is a multiple of $c.$

\section{Preliminaries}
Let $\boldsymbol{u} = (u_0, u_1, \ldots, u_{k-1})$  be a vector in $\mathbb{F}_q^k$. Then the $b-$symbol read vector $\pi_b(\boldsymbol{u})$ of $\boldsymbol{u}$ is given by:
        \[
            \pi_b(\boldsymbol{u}) = [ (u_0, u_1, \ldots, u_{b-1}), (u_1, u_{2}, \ldots, u_{b}), \ldots , (u_{k-1}, u_0, \ldots, u_{b-2}) ].
        \]
%
            For any vectors $\boldsymbol{u}, \boldsymbol{v} \in \mathbb{F}_q^k$, the $b-$symbol distance between $\boldsymbol{u}$ and $\boldsymbol{v}$ is defined as
            \[
                d_b(\boldsymbol{u},\boldsymbol{v}) = d_H(\pi_b(\boldsymbol{u}), \pi_b(\boldsymbol{v})).
            \]

     \begin{definition}
      \cite{FCBSC}
     Consider $M$ vectors $\boldsymbol{u_0}, . . . , \boldsymbol{u_{M-1}} \in {\mathbb{F}^k_q}$ and a function $f : \mathbb{F}_q^k \mapsto Im(f)$ . Then,$\ $for $t \in \mathbb{N},~ $  $\boldsymbol{B_f^{(1)}}(t, \boldsymbol{u_0}, . . . , \boldsymbol{u_{M-1}})$ and
    $\boldsymbol{B_f^{(2)}}(t, \boldsymbol{u}_0, . . . , \boldsymbol{u}_{M-1})$  are $M \times M$  $b-$symbol distance matrices with entries
 
    \begin{equation*}
        [\boldsymbol{B_f^{(1)}}(t, \boldsymbol{u}_0, . . . , \boldsymbol{u}_{M-1})]_{ij}=   
        \begin{cases}
           [2t - b + 2 - d_b(\boldsymbol{u}_i, \boldsymbol{u}_j)]^+ , & if \hspace{1mm} f (\boldsymbol{u}_i) \neq f(\boldsymbol{u}_j),\\0, & otherwise.
        \end{cases}
    \end{equation*}
    
    and 
    
    \begin{equation*}
        [\boldsymbol{B_f^{(2)}}(t, \boldsymbol{u}_0, . . . , \boldsymbol{u}_{M-1})]_{ij}=   
        \begin{cases}
           [2t + b - d_b(\boldsymbol{u}_i, \boldsymbol{u}_j)]^+ , & if \hspace{1mm} f (\boldsymbol{u}_i) \neq f(\boldsymbol{u}_j),\\0, & otherwise.
        \end{cases}
    \end{equation*}
    \end{definition}

If  $M = q^k$ then we denote  $\boldsymbol{B_f^{(1)}}(t, \boldsymbol{u}_0, . . . , \boldsymbol{u}_{M-1})$ and $ \boldsymbol{B_f^{(2)}}(t, \boldsymbol{u}_0, . . . , \boldsymbol{u}_{M-1})$ by $ \boldsymbol{B_f^{(1)}}(t) $ and $ \boldsymbol{B_f^{(2)}}(t) $ respectively.

\begin{definition}
\cite{FCBSC}
A set of codewords $\mathcal{P} = \{ \boldsymbol{p_0}, \ldots, \boldsymbol{p_{M-1}} \}$ is said to be a $\boldsymbol{B}$-irregular $b-$symbol distance code ($\boldsymbol{B_b}$-code) for some matrix $\boldsymbol{B}$ $\in \mathbb{N}_0^{M \times M}$ if there exists an ordering of the codewords in $\mathcal{P}$ such that $d_b(\boldsymbol{p_i},\boldsymbol{p_j}) \geq$ $\boldsymbol{[B]_{ij}}$ for all $i,j \in \{0, 1, \ldots, M-1\}.$ 
    \end{definition}

The smallest length $r$ for which there exists a $\boldsymbol{B}_b$-code of length $r$ is denoted by $N_b(\boldsymbol{B}).$ If all the non-diagonal entries of $\boldsymbol{B}$ are equal to a constant $D$ then we denote $N_b(\boldsymbol{B})$ by $N_b(M,D).$  


Function-correcting codes for the b-symbol read channels are defined as follows. 
\begin{definition}
\cite{FCBSC}
An encoding function, $ Enc: {\mathbb{F}^k_q} \rightarrow {\mathbb{F}_q^{k+r}},  Enc(\boldsymbol{u}) = (\boldsymbol{u},p(\boldsymbol{u}))$ defines a \textbf{function correcting $b$-symbol} code  (FCBSC) for the function $f: \mathbb{F}^k_q \rightarrow Im(f)$  if 
    \begin{equation*}  
        d_b(Enc(\boldsymbol{u_1}), Enc(\boldsymbol{u_2})) \geq 2t+1,  \hspace{3mm}\forall~ \boldsymbol{u_1}, \boldsymbol{u_2} \in {\mathbb{F}^k_q} \text{ with  } f(\boldsymbol{u_1}) \neq f(\boldsymbol{u_2})  . \\
    \end{equation*}
\end{definition}

Fig.1 illustrates the setup for FCBSCs. We will denote an FCBSC for a given function $f:\mathbb{F}_q^k \mapsto Im(f)$ and parameter $t$ by $(f,t)$-FCBSC. The smallest value of $r$ for which an $(f,t)$-FCBSC exists is called its {optimal redundancy} and is denoted by  $r_b^f(k,t).$ The optimal redundancy is bounded as follows,
\begin{theorem}
\cite{FCBSC}\label{FCBSC bound}
For any function $f: \mathbb{F}^k_q \rightarrow Im(f)$ and $\{\boldsymbol{u_1}, \ldots, \boldsymbol{u_{q^k}}\} = \mathbb{F}^k_q$, we have 
    \[
        N_b(\boldsymbol{B}_f^{(1)}(t)) \leq r_b^f(k,t) \leq N_b(\boldsymbol{B}_f^{(2)}(t)).
    \]
\end{theorem}

\section{FCBSCs for Linear Functions}
	\label{defn:linear function}

    A function $f : \mathbb{F}_{q}^{k} \to \mathbb{F}_{q}^{l} $ is said to be linear if it satisfies the following condition: $$f(\alpha \boldsymbol{u} + \beta \boldsymbol{y}) = \alpha f(\boldsymbol{u}) +\beta f(\boldsymbol{y}), \forall \boldsymbol{u},\boldsymbol{y} \in \mathbb{F}_{q}^{k} \ and \ \alpha , \beta \in  \mathbb{F}_{q}.$$ 
        
    
\noindent Such a function can always be expressed as a matrix operation $f(\boldsymbol{u}) = \boldsymbol{F}\boldsymbol{u}, \text{ for some }\boldsymbol{F} \in \mathbb{F}_{q}^{l \times k}.$ The kernel of $f$ or the null space of $\boldsymbol{F}$ is denoted by $ker(f)$. We will only consider linear functions for which $l \le k$ and $\boldsymbol{F}$ is full rank, i.e., $rank_{\mathbb{F}_{q}}(\boldsymbol{F})=l$.
We denote the range of $f$ as $Im(f) = \mathbb{F}_{q}^{l} \triangleq \{f_{0},f_{1},\hdots f_{q^{l}-1}\}$.

\subsection{Plotkin-Bound for FCBSCs for linear functions}
\begin{corollary}\label{Cor 1}
For a linear function $f:\mathbb{F}_{q}^{k} \to \mathbb{F}_{q}^{l}$, the optimal redundancy of a $t$-FCBSC is lower bounded as
	\begin{equation}
		\label{requiredbound}
		r_b^f(k,t) \ge \left(\frac{q^{b}}{q^{b}-1}\right)(2t-b+2)(1-q^{-l})-k+\left(\frac{q^{b}}{q^{b}-1}\right)\left(\frac{s}{q^{k}}\right),
	\end{equation}  
	where $s = \sum_{\boldsymbol{u}\in ker(f)}w_{b}(\boldsymbol{u}) $ i.e, the sum of Hamming weights of the vectors in $ker(f)$ and the \textbf{b-weight} of a vector $\boldsymbol{u}$ is defined as $w_b(\boldsymbol{u}) \triangleq w_{H}(\pi_b(\boldsymbol{u})).$
\end{corollary}

\begin{proof}
    Consider the set of all vectors in $\mathbb{F}_{q}^{k}$. For convenience, let $\boldsymbol{B} =  \boldsymbol{B_f^{(1)}}(t) $. Let $\mathcal{P} = \{ \boldsymbol{p_0},\boldsymbol{p_1} \ldots, \boldsymbol{p_{M-1}} \},$ with $M=q^k,$  be a $\boldsymbol{B_b}$-code of length $r$. Then $d_b(\boldsymbol{p_i},\boldsymbol{p_j}) \geq$ $\boldsymbol{[B]_{ij}}$ for all $i,j \in \{0, 1, \ldots, M-1\}.$ Summing over all $i$ and $j$, we get 
 \begin{equation}
		\label{plotkin 1}
		\sum_{i, j}^{}d_b(\boldsymbol{p}_{i},\boldsymbol{p}_{j}) \ge \sum_{i, j}^{}[\boldsymbol{B}]_{ij}.
	\end{equation}
Since $rank_{\mathbb{F}_{q}}(\boldsymbol{F})=l$, dim$_{\mathbb{F}_{q}}(ker(f)) = k-l$. Thus each coset of $ker(f)$ will contain $q^{k-l}$ elements. If $\boldsymbol{u}_i, \boldsymbol{u}_j \in \mathbb{F}_{q}^{k}$ are in the same coset of $ker(f),$ then $f(\boldsymbol{u}_i) = f(\boldsymbol{u}_j)$ and hence $ \boldsymbol{[B]_{ij}} = 0.$ Thus there will be at least $q^{k-l}$ $0$s in each column.

Suppose $\boldsymbol{u}_i, \boldsymbol{u}_j$ are in different cosets of $ker(f)$. Then 
$$\boldsymbol{[B]_{ij}} = [2t-b+2-d_b(\boldsymbol{u}_i, \boldsymbol{u}_j) ]^+. $$ Consider any other column $j'$ in $\boldsymbol{B} , j' \neq j$. Let $\boldsymbol{v}=\boldsymbol{u}_{j'}-\boldsymbol{u}_j.$ Let $i'$ be the row in $\boldsymbol{B}$ such that 
$\boldsymbol{u}_{i'} = \boldsymbol{u}_i + \boldsymbol{v}.$
Since $\pi_b(.)$ is an isomorphism from $\mathbb{F}_{q}^{k}$ to the space of all b-symbol read vectors of the vectors in $\mathbb{F}_{q}^{k}.$
\begin{align*}
     d_b(\boldsymbol{u}_{i'},\boldsymbol{u}_{j'}) &= 
     d_H(\pi_b(\boldsymbol{u}_{i'}),\pi_b(\boldsymbol{u}_{j'}) \\
&=w_b(\pi_b(\boldsymbol{u}_{j'})-\pi_b(\boldsymbol{u}_{i'})) \\
&=w_b(\pi_b(\boldsymbol{u}_{j'}-\boldsymbol{u}_{i'})) \\
&=w_b(\pi_b(\boldsymbol{u}_{j'}-\boldsymbol{u}_{i}-\boldsymbol{v})) \\
&=w_b(\pi_b(\boldsymbol{u}_{j'}-\boldsymbol{u}_{i}-(\boldsymbol{u}_{j'}-\boldsymbol{u_{j}}))) \\
&=w_b(\pi_b(\boldsymbol{u}_{j}-\boldsymbol{u}_{i})) \\
&=w_b(\pi_b(\boldsymbol{u}_{j})-\pi_b(\boldsymbol{u}_{i})) \\
& =d_H(\pi_b(\boldsymbol{u}_{i}),\pi_b(\boldsymbol{u}_{j})) \\
 &= d_b(\boldsymbol{u}_i, \boldsymbol{u}_j).
\end{align*}

Therefore $$\boldsymbol{[B]_{i'j'}} = [2t-b+2-d_b(\boldsymbol{u}_i, \boldsymbol{u}_j) ]^+. $$
Thus for every column in $\boldsymbol{B}$ each entry in the column appears in every other column. The columns of $\boldsymbol{B}$ are permutations of each other and we can write
\begin{align}
		\sum_{i, j}^{}[\mathbf{B}]_{ij} &= {\text{(no. of columns)} \times \text{(sum of one column)}}. \nonumber
\end{align}

Without loss of generality let $u_0 = \boldsymbol{0}.$ Let $I$ be the set of all vectors not in the kernel of $f$.$\forall i \in I$ We have
\begin{align*}
		[\mathbf{B}]_{i0} &\ge 2t-b+2-d_{b}(\mathbf{u}_{i},\mathbf{0})\\
		&= 2t-b+2 - w_{b}(\mathbf{u}_{i}).
        \\
       \mbox{and~~~~~~~} \sum_{i}^{} [\mathbf{B}]_{i0} &\ge (2t-b+2)(q^{k}-q^{k-l}) - \sum_{i \in I}^{} w_{b}(\mathbf{u}_{i}).
	\end{align*}
	
    The sum\footnote{For any $\boldsymbol{u} \in \mathbb{F}_q^k$ and index $i \in [k]$ the $i^{th}$ entry of its b-symbol read vector , $\pi_b(\boldsymbol{u})_i = (\boldsymbol{u}_i, \boldsymbol{u}_{i+1}, \ldots \boldsymbol{u}_{i+b-1})$ contributes a 1 to $\sum_{\boldsymbol{u} \in \mathbb{F}_{q}^{k} }^{} \boldsymbol{u}$ if $\pi_b(\boldsymbol{u})_i \neq \boldsymbol{0}$. If $\pi_b(\boldsymbol{u})_i = \boldsymbol{0}$ the remaining $k-b$ symbols can each take any of $q$ values. The number of vectors for which $\pi_b(\boldsymbol{u})_i \neq \boldsymbol{0}$ is thus $q^k -q^{k-b}$. This holds for all $i \in [k]$ $\implies \sum_{\boldsymbol{u} \in \mathbb{F}_{q}^{k} }^{} \boldsymbol{u} = k(q^k-q^{k-b})$ } of $b$-weights of all the vectors in $\mathbb{F}_{q}^{k}$ can be found to be $k(q^{k}-q^{k-b})$ for $k \geq b$ . Denoting the sum of $b$-weights of the vectors in $ker(f)$ by $s$, we have

\begin{align}
        \sum_{i}^{} [\mathbf{B}_{i0}] & \ge (2t-b+2)(q^{k}-q^{k-l}) - k(q^{k}-q^{k-b}) + s, \nonumber    
 \intertext{Thus} \nonumber
        \sum_{i.j}^{}[\mathbf{B}]_{ij} & \ge q^{k} \times \sum_{i}^{} [\mathbf{B}]_{i0} \nonumber\\
	&= {q^{k}((2t-b+2)(q^{k}-q^{k-l}) - k(q^{k}-q^{k-b}) + s)}\label{plotkin 2}.
\end{align}
For $k \geq b $, $ q^k \equiv0 (\text{mod} \ q^b) .$ From  the proof of \textbf{Lemma 4.1} in \cite{FCBSC} we have 
\begin{align}
    r \times M^2(1-q^{-b}) & \geq \sum_{i, j}^{}d_b(\boldsymbol{p}_{i},\boldsymbol{p}_{j}) \nonumber \\
    r & \geq \frac{\sum_{i, j}^{}d_b(\boldsymbol{p}_{i},\boldsymbol{p}_{j})}{q^{2k}(1-q^{-b})} \label{plotkin 3}
\end{align}
From \eqref{plotkin 1}, \eqref{plotkin 2}, \eqref{plotkin 3} and Theorem \ref{FCBSC bound} we get the required bound \eqref{requiredbound}.
\end{proof}

\subsection{Special Cases}
As  a special case, if we take $b=2$, we get the Plotkin bound for linear function-correcting codes for symbol-pair read channels over $\mathbb{F}_q$ to be 
\begin{equation*}
		r_2^f(k,t) \ge \left(\frac{q^{2}}{q^{2}-1}\right)(2t)(1-q^{-l})-k+\left(\frac{q^{2}}{q^{2}-1}\right)\left(\frac{s}{q^{k}}\right).
\end{equation*}  
Further, taking $q=2 \text{ and } b=2$ we get the Plotkin bound for function-correcting symbol pair codes (FCSPCs) defined in \cite{FCSPC}
\begin{equation}\label{plotkin 4}
		r_2^f(k,t) \ge \left(\frac{8t}{3}\right)(1-2^{-l})-k+\frac{s}{3.2^{k-2}}.
\end{equation} 
Taking $b=1$ we get
\begin{equation}\label{plotkin 5}
		r_{f}(k, t) \ge \left(\frac{q}{q-1}\right)(2t+1)(1-q^{-l})-k+\frac{s}{(q-1)(q^{k-1})},
	\end{equation}
which is exactly the Plotkin-like bound for linear function-correcting codes given in \cite{LFCC}
\subsection{Recovering Plotkin-like-Bounds for ECCs} 
If $f$ is a bijective linear function then $l=k$ and $s=0$. Substituting these in Corollary \ref{Cor 1} we get
\begin{equation*}
    n \triangleq k + r_b^f(k,t) \ge \left(\frac{q^{b}}{q^{b}-1}\right)(2t-b+2)(1-q^{-k}),
\end{equation*}
which matches the Plotkin-like bound for ECCs for b-symbol read channels given in \cite{Plotkin-b-symbol}.\\
Substituting $l=k$ and $s=0$ in equation \eqref{plotkin 4} we get 
\begin{equation*}
    n \triangleq k + r_2^f(k,t) \geq \left(\frac{8t}{3}\right)(1-2^{-k}),
\end{equation*}
which matches the Plotkin-like bound for ECCs for the symbol-pair read channel given in \cite{Plotkin-symbol-pair}.
Making the same substitutions in equation \eqref{plotkin 5} we recover the Plotkin bound for ECCs,
\begin{equation*}
		n \triangleq k+r_{f}(k, t) \ge \left(\frac{q}{q-1}\right)(2t+1)(1-q^{-k}).
\end{equation*}

\end{document}